\documentclass{eptcs} 

\providecommand{\event}{}

\providecommand{\firstpage}{1}

\usepackage{amssymb}
\usepackage{url}
\usepackage{times}
\usepackage{epsfig}
\usepackage{breakurl}
\usepackage{color}

\makeatletter

\def\hmo{{\cal O}}

\def\hml{{\it HML}}

\def\eqhmo{\sim_{\hmo}}

\def\hmeq{\sim_{\cal O}}

\def\enc{\partial}

\def\transa{\stackrel{a}{\rightarrow}}

\def\transa{\stackrel{a}{\rightarrow}}
\def\transb{\stackrel{b}{\rightarrow}}

\def\iff{\Leftrightarrow}
\def\iffdef{\stackrel{\rm def}\Leftrightarrow}

\def\implies{\Rightarrow}
\def\true{{\sf T}}
\def\false{{\sf F}}

\newtheorem{defi}{Definition}
\newtheorem{theo}{Theorem}
\newtheorem{prop}{Proposition}
\newtheorem{lemm}{Lemma}
\newtheorem{coro}{Corollary}

\newenvironment{theorem}{\begin{theo} \rm }{\end{theo}}

\newenvironment{lemma}{\begin{lemm} \rm }{\end{lemm}}

\newenvironment{proof}{\begin{trivlist} \item[\hspace{\labelsep}\bf Proof:]}{\hfill $\Box$ \end{trivlist}}

\def\foralli{\forall_{i \in I}}

\newcommand{\diam}[1]{\langle#1\rangle}
\newcommand{\lmerge}[2]{#1 ||_{\_} #2}

\title{Congruence from the Operator's Point of View:\\
Compositionality Requirements on Process Semantics}
\author{Maciej Gazda \& Wan Fokkink
\institute{Vrije Universiteit\\
Department of Computer Science\\
De Boelelaan 1081a, 1081 HV Amsterdam, Netherlands}
\email{m.w.gazda@student.vu.nl,~wanf@cs.vu.nl}
}
\def\authorrunning{M. Gazda and W. Fokkink}
\def\titlerunning{Compositionality Requirements on Modal Characterizations of Process Semantics}

\begin{document}
\maketitle
\begin{abstract}
One of the basic sanity properties of a behavioural semantics is that it constitutes
a congruence with respect to standard process operators. This 
issue has been traditionally addressed by the development of rule formats 
for transition system specifications that define process algebras. In 
this paper we suggest a novel, orthogonal approach. Namely, we focus on a 
number of process operators, and for each of them attempt to find the 
widest possible class of congruences. To this end, we impose restrictions 
on sublanguages of Hennessy-Milner logic, so that a semantics whose 
modal characterization satisfies a given criterion is guaranteed to be 
a congruence with respect to the operator in question. We investigate action 
prefix, alternative composition, two restriction operators, and parallel composition.
\end{abstract}

\section{Introduction}

Congruence is one of the most important properties of a behavioural semantics. The reason is that the fundamental issue in process algebra - providing sound and complete axiomatisations for collections of process operators - requires that these operators are compositional. Only then we can use equational logic priciples and provide sound axioms.

There is a large amount of research to find ways of ensuring the congruence property. The basic methodology is to impose restrictions on operator definitions; there is a notion of a rule format for transition system specifications which provide operational semantics for process algebras. If a process operator is defined with rules that fit within a format, then the semantics in question is a congruence with respect to this operator. Examples include the panth format for bisimulation semantics \cite{Ver95} and formats designed specifically for several decorated trace semantics \cite{BlFoGl04}. The focus here is on \textit{semantics}; rule formats are most often defined with one particular process semantics in mind. Interestingly, in \cite{BlFoGl04}, the modal characterization of a process semantics is taken as starting point to derive the syntactic constraints of the congruence format for this semantics. A modal characterization of a semantics is a sublanguage of Hennessy-Milner logic such that two processes are semantically equivalent if and only if they satisfy exactly the same formulas in the modal characterization of the semantics. For almost all process semantics in van Glabbeek's spectrum \cite{Gla01} there is a corresponding modal characterization.
 
In this paper, we attempt to look at the compositionality issue from an \textit{operator's} point of view. For a number of basic process operators, we determine conditions that a process semantics should satisfy in order to be congruence with respect to such an operator. To be more precise, given a process operator, we develop syntactic constraints on modal characterizations; if the modal characterization of a process semantics satisfies these constraints, then the process operator is guaranteed to be compositional with respect to this semantics. So instead of going from a process semantics to a class of transition system specifications for which that semantics is a congruence, we go from the transition rules of a process operator to a class of process semantics for which this operator is compositional. This approach gives us an orthogonal view on compositionality, and provides further insight into connections between process algebra and modal logic.

\section{Preliminaries}

We work in the usual setting of labelled transition systems (LTSs), which consist of a set $S$ of states $p$ (also called processes), a set ${\it Act}$
of actions $a$, and a set of transitions $p \transa p'$.

\subsection{Hennessy-Milner logic}
\label{sec:hml}

Hennessy-Milner logic ($\hml$) \cite{HeMi85} is a modal logic for specifying properties of states in an LTS.
There exist different versions of $\hml$ \cite{HeMi85,Gla01,BlRiVe01}. The choice of syntax is important here,
even if two logics have the same expressivity; compositionality requirements established for some version of $\hml$
(e.g.\ with diamond, conjunction and negation only) may become insufficient when we add other operators
(e.g.\ box), because these extra operators may require syntactic requirements of their own.
Our point of departure is the infinitary $\hml$ variant without box and disjunction. The $\hml$ syntax is therefore as follows:
\begin{center}
$
\varphi ~~::=~~ \true ~\mid~ \bigwedge_{i\in I}\varphi_i ~\mid~ \diam{a}\varphi ~\mid~ \neg\varphi
$
\end{center}
where $I$ is an arbitrary index set, and $a$ ranges over the set {\it Act} of actions.
Furthermore, we use $\false$ as an abbreviation for $\neg\true$. We introduce some additional notations, based on the standard notion of context. A context, notation $C[]$, is a $\hml$ formula with one occurrence of $[]$. A \textit{multicontext $C[]_{i \in I}$} is a $\hml$ formula containing one or more $[]$ symbols, indexed by the elements from $I$. For a (multi)context $C[\varphi_i]_{i \in I}$, a formula is obtained by replacing the $[]_i$ symbols with formulas $\varphi_i$.
Finally, we introduce an \textit{$n$-level context}, which means that the context symbol has $n$ diamond operators above it. It is defined inductively as follows:
 \begin{itemize}
 \item $[]$ is a ${\bf 0}$-level context;
 \item if $C_{n}[]$ is an $n$-level context, then $\neg C_{n}[]$ and $C_{n}[] \wedge \bigwedge_{i \in I} \varphi_i$ are $n$-level contexts;
 \item if $C_{n}[]$ is an $n$-level context, then $\diam{a}C_{n}[]$ is an ($n+1$)-level context.
 \end{itemize}
An example of a ${\bf 0}$-level context is $\diam{a}\diam{b} \true \wedge \neg []$, while $\diam{a} (\diam{a}\diam{b} \true \wedge [])$ is a $1$-level context.

A sublanguage $\hmo$ of $\hml$ gives rise to a process equivalence by identifying those processes which satisfy exactly the same formulas from $\hmo$:
\begin{center}
$p \eqhmo q \,\iffdef\, \forall_{\varphi \in \hmo}: (p \models \varphi \iff q \models \varphi)$.
\end{center}
We call $\hmo$ a  modal characterization of $\hmeq$. Below, examples of modal characterizations of standard process equivalences from the literature are given (see \cite{Gla01}):
\begin{itemize}
	\item \textit{trace observations}:\\
	 $\hmo_{T} \hspace{5.0 pt} \varphi ::= \true \hspace{3.0 pt}|\hspace{3.0 pt} \diam{a} \varphi '~(\varphi' \in \hmo_{T})$
	\item \textit{completed trace obervations}:\\
	 $\hmo_{CT} \hspace{5.0 pt} \varphi ::= \true \hspace{3.0 pt}|\hspace{3.0 pt} \diam{a} \varphi '~(\varphi' \in \hmo_{CT})\hspace{3.0 pt}|\hspace{3.0 pt} \bigwedge_{a \in {\it Act}} \neg \diam{a} \true$
	\item \textit{failures observations}:\\
	 $\hmo_{F} \hspace{5.0 pt} \varphi ::= \true \hspace{3.0 pt}|\hspace{3.0 pt} \diam{a} \varphi '~(\varphi' \in \hmo_{F})\hspace{3.0 pt}|\hspace{3.0 pt} \bigwedge_{i \in I} \neg \diam{a_i} \true$
  \item \textit{readiness observations}:\\	 
	$\hmo_{R} \hspace{5.0 pt} \varphi ::= \true \hspace{3.0 pt}|\hspace{3.0 pt} \diam{a} \varphi '~(\varphi' \in \hmo_{R})\hspace{3.0 pt}|\hspace{3.0 pt} \bigwedge_{i \in I} \neg \diam{a_i} \true \wedge \bigwedge_{j \in J} \diam{b_j} \true$
	\item \textit{failure trace observations}:\\	 
	$\hmo_{FT} \hspace{5.0 pt} \varphi ::= \true \hspace{3.0 pt}|\hspace{3.0 pt} \diam{a} \varphi '~(\varphi' \in \hmo_{FT})\hspace{3.0 pt}|\hspace{3.0 pt} \bigwedge_{i \in I} \neg \diam{a_i} \true \wedge \varphi '~(\varphi ' \in \hmo_{FT})$
	\item \textit{ready trace observations}:\\	
	$\hmo_{RT} \hspace{5.0 pt} \varphi ::= \true \hspace{3.0 pt}|\hspace{3.0 pt} \diam{a} \varphi '~(\varphi' \in \hmo_{RT})\hspace{3.0 pt}|\hspace{3.0 pt} \bigwedge_{i \in I} \neg \diam{a_i} \true \wedge \bigwedge_{j \in J} \diam{b_j} \true \wedge \varphi '~(\varphi ' \in \hmo_{RT})$
	\item \textit{simulation observations}:\\	
	$\hmo_{1S} \hspace{5.0 pt} \varphi ::= \true \hspace{3.0 pt}|\hspace{3.0 pt} \diam{a} \varphi '~(\varphi' \in \hmo_{1S}) \hspace{3.0 pt}|\hspace{3.0 pt} \bigwedge_{i \in I} \varphi_{i}~(\varphi_{i} \in \hmo_{1S})$
	\item \textit{ready simulation observations}:\\	 
	$\hmo_{RS} \hspace{5.0 pt} \varphi ::= \true \hspace{3.0 pt}|\hspace{3.0 pt} \diam{a} \varphi '~(\varphi' \in \hmo_{RS}) \hspace{3.0 pt}|\hspace{3.0 pt} \neg \diam{a} \true \hspace{3.0 pt}|\hspace{3.0 pt} \bigwedge_{i \in I} \varphi_{i}~(\varphi_{i} \in \hmo_{RS})$
	\item \textit{$n$-nested simulation observations for $n \geq 2$}:\\	 
	$\hmo_{nS} \hspace{5.0 pt} \varphi ::= \true \hspace{3.0 pt}|\hspace{3.0 pt} \diam{a} \varphi '~(\varphi' \in \hmo_{nS}) \hspace{3.0 pt}|\hspace{3.0 pt} \bigwedge_{i \in I} \varphi_{i}~(\varphi_{i} \in \hmo_{nS}) \hspace{3.0 pt}|\hspace{3.0 pt} \neg \varphi '~(\varphi' \in \hmo_{(n-1)S})$
	\item \textit{ bisimulation observations}:\\
	$\hmo_{B} \hspace{5.0 pt} \varphi ::= \true \hspace{3.0 pt}|\hspace{3.0 pt} \diam{a} \varphi '~(\varphi' \in \hmo_{B}) \hspace{3.0 pt}|\hspace{3.0 pt} \bigwedge_{i \in I} \varphi_{i}~(\varphi_{i} \in \hmo_{B})\hspace{3.0 pt}|\hspace{3.0 pt} \neg \varphi '~(\varphi ' \in \hmo_{B}) $
	\end{itemize}

\vspace{2mm}
We write $\varphi\equiv\varphi'$ if $p\models\varphi\Leftrightarrow p\models\varphi'$ for any process $p$ in any LTS.
Given an $\hmo\subseteq\hml$, we write $\hmo^\equiv$ for the set of HML formulas $\varphi$ for which there
exists a $\varphi'\in\hmo$ with $\varphi\equiv\varphi'$.

\subsection{BCCSP}

Any LTS isomorphic with a finite tree can be described with the following process algebra BCCSP, consisting of three operators:
	
\begin{itemize}
\item a nullary process ${\bf 0}$ which does not have any behaviour;
\item action prefix $a.()$ for $a\in{\it Act}$: a unary operator which represents execution of a single action followed by the process given as the argument, defined by the transition rule
\begin{center}
\[ \frac{}{ax \transa x} \]
\end{center}
\item alternative composition ($+$), a nondeterministic choice between two processes, defined by the transition rules
\begin{center}
\[ \frac{x \transa x'}{x + y \transa x'} 
~~~~~~~~~~ \frac{y \transa y'}{x + y \transa y'}\]
\end{center}
\end{itemize}
	
In this paper we focus on several process operators from the literature, and try to establish which syntactic properties a modal language $\hmo \subseteq \hml$ should satisfy to guarantee that the induced equivalence is a congruence with respect to the given operator. That is, given a process operator $f$, we will search for a syntactic condition $C$ such that if $\hmo$ satisfies $C$, then $f$ is compositional with respect to $\eqhmo$.

\section{Basic operators}

\subsection {Alternative composition}

We start with alternative composition, which expresses a nondeterministic choice between two processes. We want to find a general property of a modal language that would guarantee congruence of the induced equivalence with respect to alternative composition. Our first observation is that the behaviour of an alternative composition $p_1+p_2$ \textit{after performing the first step} is completely determined by the behaviour of one of the components. For example, $p_1+p_2 \models \diam{a}\varphi$ if and only if either $p_1 \models \diam{a} \varphi$ or $p_2 \models \diam{a} \varphi$. The only potential problem can occur when there is a formula with a conjunction at level 0 (i.e., not in the scope of an action prefix). For instance, consider $ \hmo = \{ \diam{a} \true \wedge \diam{b} \true \} $.  We have $a{\bf 0} \hmeq {\bf 0}$ and $b{\bf 0} \hmeq {\bf 0}$, but $a{\bf 0} + b{\bf 0} \not\hmeq {\bf 0} + {\bf 0}$. As it turns out, it suffices to simply close the language on sub-conjunctions at level 0.
  
\begin{theorem} 
Let $\hmo \subseteq \hml$. If for any 0-level context $C_{0}[]$ and $\varphi_i\in \hml$ for $i\in I$,
\begin{center}
(AC)~~~~~$C_{0}[\bigwedge_{i \in I} \varphi_i]\in\hmo$ implies that $\foralli: (\varphi_i \in \hmo^\equiv)$,
\end{center}
then $\hmeq$ is a congruence with respect to alternative composition ($+$).
\end{theorem}
\begin{proof} Assume a modal language $\hmo$ with the AC property. Let $p_1 \hmeq q_1$ and $p_2 \hmeq q_2$. We show that for any $\varphi\in \hmo$:
\begin{center}
$p_1 + p_2 \models \varphi ~\implies~ q_1 + q_2 \models \varphi$
\end{center}
(the converse implication "$\Leftarrow$" is symmetric).  We apply induction on the structure of $\varphi$. The base case ($\true$) is trivial. We proceed with the inductive step. Assume that $p_1 + p_2 \models \varphi$. We have to consider the following cases:
 
\begin{itemize}
\item $\varphi = \diam{a} \psi$: then either $p_1 \models \diam{a} \psi$ or $p_2 \models \diam{a} \psi$. From the equivalence of components we have either $q_1 \models \diam{a} \psi$ or $q_2 \models \diam{a} \psi$, which yields $q_1 + q_2 \models \diam{a} \psi$.

\item $\varphi = \bigwedge_{i \in I} \varphi_i$: we have $p_1 + p_2 \models \bigwedge_{i \in I}\, \varphi_i ~\iff~ \forall_{i \in I}: p_1 + p_2 \models \varphi_i ~\iff~ \forall_{i \in I}: q_1 + q_2 \models \varphi_i$ (AC + inductive hypothesis) $ \iff~ q_1 + q_2 \models \bigwedge_{i \in I} \varphi_i$.

\item $\varphi = \neg \psi$: let $\psi'$ be the outermost subformula of $\varphi$ which does not begin with a "$\neg$" symbol (so $\varphi = (\neg)^{n} \psi'$). Then $\varphi$ is logically equivalent to either $\psi'$ or $\neg \psi'$. The case $\psi' = \true$ is trivial. Also, the case where $\varphi \equiv \psi'$ can be handled analogously as the first two cases. We thus have to consider two possibilities:

\begin{itemize}
\item $\varphi \equiv \neg \diam{a} \varphi'$: we have $p_1 + p_2 \models \neg \diam{a} \varphi' ~\iff~ p_1 \models \neg \diam{a} \varphi' \wedge p_2 \models \neg \diam{a} \varphi' ~\iff~ q_1 \models \neg \diam{a} \varphi' \wedge q_2 \models \neg \diam{a} \varphi'$ (equivalence of components) $\iff~q_1 + q_2 \models \neg \diam{a} \varphi'$.

\item $\varphi \equiv \neg \bigwedge_{i \in I} \varphi_i$: we have $p_1 + p_2 \models \neg \bigwedge_{i \in I}\, \varphi_i ~\iff~ \exists_{i \in I}: p_1 + p_2 \models \neg \varphi_i ~\iff~ \exists_{i \in I}: q_1 + q_2 \models \neg \varphi_i$ (AC + inductive hypothesis) $ \iff~ q_1 + q_2 \models \neg \bigwedge_{i \in I} \varphi_i  $.
\end{itemize}

\end{itemize}
 
\end{proof}

For example, consider $\hmo = \{  \diam{a}(\diam{a}\true \wedge \diam{b}\true)  \wedge \neg \diam{b}\true, \diam{a}(\diam{a}\true \wedge \diam{b}\true), \neg \diam{b}\true \}$. The language $\hmo$ satisfies the AC requirement, and so the corresponding equivalence $\hmeq$ is a congruence with respect to $+$. 

\vspace{2mm}

Almost all modal characterizations of standard process semantics from Section \ref{sec:hml} fulfill AC. The only exception is the modal characterization of completed trace equivalence,
although we can provide an alternative characterization that meets the AC requirement:
\begin{center}
$\hmo^\ast_{CT} \hspace{5.0 pt} \varphi ::= \varphi'~(\varphi' \in \hmo_{CT})\hspace{3.0 pt}|\hspace{3.0 pt} \neg \diam{a} \true$ 
\end{center}
The characterization $\hmo^\ast_{CT}$ is the same as $\hmo_{CT}$, except that it includes formulas $\neg \diam{a} \true$. Clearly this does not change the corresponding semantics.

\subsection{Action prefix}

In the case of action prefix, it is easy to obtain a sufficient congruence requirement; the crucial observatioin is that $a.p \models \diam{a} \varphi$ if and only if $p \models \varphi$, so we need to make sure that for each formula $C_0[\diam{a}\varphi] \in \hmo$, the subformula $\varphi$ also belongs to the language $\hmo$. If this is not the case, an equivalence might not be a congruence. For instance, if $\hmo = \{ \diam{a}\diam{a}\true \}$, then $a.{\bf 0} \hmeq {\bf 0}$, but $a.a.{\bf 0} \not\hmeq a.{\bf 0}$.

\begin{theorem} Let $\hmo \subseteq \hml$ and fix $a \in {\it Act}$. If for any 0-level context $C_0[]$ and $\varphi\in \hml$,
\begin{center}
(AP)~~~~~$C_0[\diam{a} \varphi] \in \hmo$ implies that $\varphi \in \hmo^\equiv$,
\end{center}
then $\hmeq$ is a congruence with respect to the action prefix operator $a.()$.
 \end{theorem}
\begin{proof}
Let $p \hmeq q$. We need to show that for any $\varphi \in \hmo$, $a.p \models \varphi \iff a.q \models \varphi$.

Take any $\varphi \in \hmo$. Let $\varphi = C[\diam{a_i} \varphi_i]_{i \in I}$ such that the multicontext $C[]_{i \in I}$ does not contain any action prefix symbols. That is, the $\diam{a_i} \varphi_i$ for $i\in I$ are all action prefix subformulas of $\varphi$ that appear at level zero. Since $C[]_{i\in I}$ is built from only $\true$, conjunction and negation, whether a process satisfies $\varphi$ is completely determined by the satisfiability of $\diam{a_i} \varphi_i$ for $i \in I$ by this process. In other words, if $\forall_{i \in I}: (p_1 \models \diam{a_i} \varphi_i \iff p_2 \models \diam{a_i} \varphi_i)$, then $p_1 \models \varphi \iff p_2\models \varphi$. Coming back to our setting with $a.p$ and $a.q$, take an arbitrary $i \in I$. We have:\\
$a.p \models \diam{a_i} \varphi_i$\\
$\iff (a_i = a) \wedge p \models \varphi_i$\\
$\iff (a_i = a) \wedge q \models \varphi_i$ (AP + $p\hmeq q$)\\
$\iff a.q \models \diam{a_i} \varphi_i$.\\
The choice of $i$ was arbitrary, hence the earlier remark yields: $a.p \models \varphi \iff a.q \models \varphi$.
\end{proof}
The AP condition is satisfied by all modal characterizations from Section \ref{sec:hml}.

\section{Restriction operators: projection and encapsulation}

We now consider projection and encapsulation operators. The $n$th projection of a process $p$, for $n\geq 0$, mimicks the behaviour of $p$ up to level $n$:
\begin{center}
\[ \frac{x \stackrel{a}{\rightarrow}x'}{\pi_{n+1}(x)\stackrel{a}{\rightarrow}\pi_{n}(x')} \]
\end{center}
Applying encapsulation with parameter $B \subseteq {\it Act}$ removes all transitions whose labels are in $B$ from the process:
\begin{center}
\[ \frac{x \stackrel{a}{\rightarrow}x'~~(a \not \in B)}{\enc_{B}(x)\stackrel{a}{\rightarrow}\enc_{B}(x')} \]
\end{center}
More generally, we consider unary restriction operators $f$ such that given a process $p$, the process $f(p)$ can be viewed as a subgraph of $p$. Below we will give a precise description of which restriction operators are covered. For the projection operator $\pi_n$ as well as for the encapsulation operator $\enc_B$, given any $\hml$ formula we can deduce in advance which of its subformulas $\diam{a} \varphi$ will always yield false, regardless of the process $\pi_n(p)$ or $\enc_B(p)$ for which the $\hml$ formula is evaluated. In case of a process $\pi_n(p)$, any subformula $\diam{a} \varphi$ that appears at level $n$ can be replaced by $\false$. And in case of a process $\enc_B(p)$, any subformula $\diam{b} \varphi$ with $b \in B$ can be replaced by $\false$.

We cannot reason in this way about any restriction operator. For example, consider the priority operator $\theta$, which assumes a partial order $<$ on the set of actions and allows us to execute an action only if no action with higher priority is executable at the same time:
\begin{center}
\[ \frac{x \transa x' ~~~~~~ \forall b\in{\it Act}\,(a<b\,\Rightarrow\, x \not \transb)} {\theta(x) \transa \theta(x')} \]
\end{center}
Suppose that $a > b$ and there is a process $p$ of which we only know that it satisfies $\diam{b} \true$. This knowledge is not sufficient to determine whether $\theta(p) \models \diam{b} \true$.

Let $f$ be a unary operator such as $\pi_n$ or $\enc_B$. We would like to define for each formula $\varphi \in \hml$ a corresponding formula ${\it cut}_f(\varphi)$ in which every subformula $\diam{a} \varphi'$ which is known in advance to be unsatisfiable when evaluating any process $f(p)$ is replaced by $\false$. Actually this means that either we can replace a larger subformula by $\true$, or the entire formula becomes $\false$. Namely, we can replace the first innermost negation symbol (closest to the introduced $\false$) and the following subformula by $\true$; if the $\false$ symbol does not appear within the scope of a negation symbol, then the whole formula yields $\false$. If a language $\hmo\subseteq\hml$ is closed under ${\it cut}_f$, then it induces a congruence with respect to $f$. The whole idea is made formal below.

\begin{lemma}
\label{lem:cut1} Let $f$ be a unary process operator. Suppose there exists a function  ${\it cut}_f: \hml \longrightarrow \hml$ such that for any process $p$ and $\varphi\in\hml$,
\begin{center}
\begin{tabular}{l l}
(CUT)~~ & $f(p) \models \varphi~\iff~p \models {\it cut}_f(\varphi)$\\
\end{tabular}
\end{center}
Then for any language $\hmo$ satisfying
\begin{center}
$\varphi \in \hmo~\implies~({\it cut}_f(\varphi) \in \hmo^\equiv~\lor~{\it cut}_f(\varphi)\equiv \false)$
\end{center}
the corresponding equivalence $\hmeq$ is a congruence with respect to $f$.
\end{lemma}
\begin{proof} Suppose $\hmo \subseteq \hml$ and $p \hmeq q$. We have
$f(p) \models \varphi ~\iff~ p \models {\it cut}_f(\varphi)$ (CUT) $\iff q \models {\it cut}_f(\varphi)$ (either because ${\it cut}_f(\varphi) \in \hmo^\equiv$ and $p \hmeq q$, or because ${\it cut}_f(\varphi)\equiv \false$) $\iff f(q) \models \varphi$ (CUT).
\end{proof}
The next lemma gives an explicit condition for a modal language to induce a congruence in case ${\it cut}_f$ formulas are obtained from the original ones by turning certain subformulas $\diam{a} \varphi$ into $\false$.

\begin{lemma}
\label{lem:cut2}
Assume $f$ and ${\it cut}_f$ are as in Lem.~\ref{lem:cut1}, and satisfy CUT. Suppose that for each $\varphi \in \hml$ there exists a multicontext $C[]_{ i \in I}$ such that 
$\varphi = C[\diam{a_i} \varphi_i]_{i \in I}$ and ${\it cut}_f(\varphi) \equiv C[\false]_{i \in I}$.
Then for each language $\hmo \subseteq \hml$ that satisfies for any context $C'[]$ and $\varphi \in \hml$,
\begin{center}
\begin{tabular}{l l}
(RES) & $C'[\neg \varphi] \in \hmo$ implies $C'[\true] \in \hmo^\equiv$,
\end{tabular}
\end{center}
the corresponding equivalence $\hmeq$ is a congruence with respect to $f$.
\end{lemma}
\begin{proof} By Lem.\ \ref{lem:cut1} it suffices to prove that for all $\varphi \in \hmo$ either ${\it cut}_f(\varphi) \in \hmo$ or ${\it cut}_f(\varphi) \equiv \false$. Take any $\varphi \in \hmo$ such that ${\it cut}_f(\varphi) \not \equiv \false$. By assumption, ${\it cut}_f(\varphi) \equiv C[\false]_{i \in I}$ for some multicontext $C[]_{i \in I}$. Since ${\it cut}_f(\varphi) \not \equiv \false$, clearly each occurrence of $\false$ in this formula must be within the scope of a negation symbol.
Hence ${\it cut}_f(\varphi) \equiv C'[\neg D^i[\false]]_{i \in I}$, where we can choose contexts $D^i[]$ for $i\in I$ such that in each $D^i[]$, $[]$ is not within the scope of a negation.
Then $C'[\neg D^i[\false]]_{i \in I} \equiv C'[\true]_{i \in I}$. Since $\hmo$ satisfies RES, $C'[\true]_{i \in I} \in \hmo^\equiv$. Hence ${\it cut}_f(\varphi) \in \hmo^\equiv$.
\end{proof}
We have provided a compositionality framework for a general class of restriction operators. What remains is to provide ${\it cut}_f$ functions for the projection and encapsulation operators.

\begin{lemma} 
\label{lem:cut3}
The functions ${\it cut}_f$ defined below are proper cutting functions (i.e., they satisfy condition CUT of Lem.~\ref{lem:cut1}).\vspace{4mm}\\
a) For the projection operators $\pi_n$ with $n\geq 0$:\vspace{2mm}\\
\begin{tabular}{lll}
${\it cut}_n(\true) = \true$ &
${\it cut}_n(\bigwedge_{i \in I} \varphi_i) = \bigwedge_{i \in I} {\it cut}_n(\varphi_i)$ &
${\it cut}_n(\neg \varphi) = \neg {\it cut}_n(\varphi)$ \vspace{2mm} \\
${\it cut}_0(\diam{a} \varphi) = \false$ & ${\it cut}_{n+1}(\diam{a}\varphi) = \diam{a}{\it cut}_n(\varphi)$
\end{tabular}
\\ \\ \\
b) For the encapsulation operators $\enc_B$ with $B \subseteq {\it Act}$:\vspace{2mm}\\
\begin{tabular}{lll}
${\it cut}_B(\true) = \true$ &
${\it cut}_B(\bigwedge_{i \in I} \varphi_i) = \bigwedge_{i \in I} {\it cut}_B(\varphi_i)$ &
${\it cut}_B(\neg \varphi) = \neg {\it cut}_B(\varphi)$ \vspace{2mm} \\
${\it cut}_B(\diam{a} \varphi) = \false$ if $a \in B$ &
${\it cut}_B(\diam{a}\varphi) = \diam{a}{\it cut}_B(\varphi)$ if $a \not \in B$
\end{tabular}
\end{lemma}
\begin{proof}
a) We prove CUT by induction on the structure of $\varphi$.
\begin{itemize}
\item $\varphi = \true$:

$\pi_n(p) \models \true$  and $p\models {\it cut}_n(\true)=\true$.

\item $\varphi = \diam{a} \psi$:

We distinguish the cases $n=0$ and $n > 0$.
Clearly $\pi_0(p) \not\models \diam{a} \psi$ and $p\not\models {\it cut}_0(\diam{a} \psi)=\false$.

If $ n > 0$, then $\pi_{n}(p) \models \diam{a} \psi
\iff \exists p': p \transa p' \wedge \pi_{n-1}(p') \models \psi$ (transition rule for $\pi_n$)
$\iff \exists p': p \transa p' ~\wedge~ p' \models {\it cut}_{n-1}(\psi)$ (structural induction)
$\iff p \models \diam{a}{\it cut}_{n-1}(\psi)
\iff p \models {\it cut}_{n}(\diam{a} \psi)$ (definition of ${\it cut}_n$).

\item $\varphi = \bigwedge_{i \in I} \psi_i$:

$\pi_n(p) \models \bigwedge_{i \in I} \psi_i
\iff \foralli: \pi_n(p) \models \psi_i
\iff \foralli: p \models {\it cut}_n(\psi_i)$ (structural induction)
$\iff p \models {\it cut}_n(\bigwedge_{i \in I} \psi_i)$ (definition of ${\it cut}_n$).

\item $\varphi = \neg \psi$:

$\pi_n(p) \models \neg \psi
\iff \pi_n(p) \not \models \psi
\iff p \not \models {\it cut}_n(\psi)$ (structural induction)
$\iff p \models \neg {\it cut}_n(\psi)$ $\iff p \models {\it cut}_n(\neg \psi)$ (definition of ${\it cut}_n$).
\end{itemize}


b) Again we use structural induction on $\varphi$.
\begin{itemize}
\item $\varphi = \true$:

$\enc_B(p) \models \true$  and $p \models {\it cut}_B(\true)=\true$.

\item $\varphi = \diam{a} \psi$:

Suppose first that $a  \in B$. Then $\enc_B(p) \not \models \diam{a} \psi$ (transition rule for $\enc_B$)
and $p \not \models {\it cut}_B(\diam{a} \psi)=\false$ (definition of ${\it cut}_B$).

Suppose now that $a \not \in B$.
Then $\enc_B(p) \models \diam{a} \psi$
$\iff \exists p': p \transa p' \wedge \enc_B(p') \models \psi$ (transition rule for $\enc_B$)
$\iff \exists p': p \transa p' \wedge p' \models {\it cut}_B(\psi)$ (structural induction)
$\iff p \models \diam{a}{\it cut}_B(\psi) \iff p \models {\it cut}_B(\diam{a}\psi)$ (definition of ${\it cut}_B$).

\item $\varphi = \bigwedge_{i \in I} \psi_i$:

$\enc_B(p) \models \bigwedge_{i \in I} \psi_i
\iff \foralli: \enc_B(p) \models \psi_i
\iff \foralli: p \models {\it cut}_B(\psi_i)$ (structural induction)
$\iff p \models {\it cut}_B(\bigwedge_{i \in I} \psi_i)$ (definition of ${\it cut}_B$).

\item $\varphi = \neg \psi$:

$\enc_B(p) \models \neg \psi$ ~
$ \iff \enc_B(p) \not \models \psi
\iff p \not \models {\it cut}_B(\psi)$ (structural induction)
$\iff p \models \neg {\it cut}_B(\psi)$ $\iff p \models {\it cut}_B(\neg \psi)$ (definition of ${\it cut}_B$).
\end{itemize}

\end{proof}
\begin{theorem}
For any language $\hmo \subseteq \hml$ satisfying RES, the corresponding equivalence $\hmeq$ is a congruence with respect to the projection operators $\pi_n$ (for $n\geq 0$) and the encapsulation operators $\enc_B$ (for $ B \subseteq {\it Act}$).
\end{theorem}

\begin{proof}
By Lem.\ \ref{lem:cut3}, the functions ${\it cut}_n$ and ${\it cut}_B$ satisfy CUT.
Observe that the ${\it cut}_n$ and ${\it cut}_B$ functions defined in the Lem.~\ref{lem:cut3} only replace certain subformulas $\diam{a} \psi$ of the original formula with $\false$. So they meet the requirements of Lem.~\ref{lem:cut2}. Congruence is thus an immediate consequence of Lem.\ \ref{lem:cut2}.
\end{proof}

To demonstrate that the RES requirement is essential, consider the following counterexamples.
\begin{itemize}
\item
For projection, take $\hmo = \{ \neg \diam{a} \neg \diam{a} \true \}$. We have $aa{\bf 0} \hmeq {\bf 0}$, but $\pi_1(aa{\bf 0}) \not\hmeq \pi_1({\bf 0})$.
\item
 For encapsulation, take $\hmo = \{ \diam{a} \neg \diam{b} \true\}$. We have $ab{\bf 0} \hmeq {\bf 0}$, but $\enc_{\{b\}}(ab{\bf 0}) \not \hmeq \enc_{\{b\}}({\bf 0})$.
\end{itemize}

The RES requirement is satisfied by every characterization from Section \ref{sec:hml}, except for completed trace observations. Completed trace equivalence is a congruence with respect to projection operators, but not encapsulation. Take for instance the completed trace equivalent processes $a(b{\bf 0}+c{\bf 0})$ and $ab{\bf 0}+ac{\bf 0}$. We have $\enc_{\{b\}}(a(b{\bf 0}+c{\bf 0})) \,\sim_{CT}\, ac{\bf 0} \,\not \sim_{CT}\, a{\bf 0}+ac{\bf 0} \,\sim_{CT}\, \enc_{\{b\}}(ab{\bf 0}+ac{\bf 0})$.

\section{Parallel composition ($||$)}

We now consider the parallel composition operator (without communication). That is, $p ||q$ behaves as $\lmerge{p}{q} + \lmerge{q}{p}$ where the left-merge operator is defined by
\begin{center}
\[ \frac{x \transa x'}{\lmerge{x}{y} \transa x' || y} \]
\end{center}

Let us restrict for a moment to only trace formulas (meaning that conjunctions are disregarded). The following example shows that the requirement AP and even being closed under substrings is not sufficient (by a substring of $w$ we mean a subsequence constisting of elements appearing \textit{consecutively} in $w$). 
Take $\hmo = \{ \diam{a}\true,\diam{b}\true,\diam{a}\diam{b}\true,\diam{a}\diam{b}\diam{a}\true,\diam{b}\diam{a}\true \}$. This language not only satisfies AP, but is also closed under prefixes and substrings (but not arbitrary subsequences). However, we have $aa{\bf 0} \hmeq a{\bf 0}$, but $aa{\bf 0} || b{\bf 0} \models \diam{a}\diam{b}\diam{a}\true$ while $a{\bf 0}||b{\bf 0}$ does not satisfy this formula.

This example suggests that if a trace $\sigma$ belongs to $\hmo$, then all \textit{subsequences} of $\sigma$ must belong to the language as well. This is not unexpected; the behaviour of parallel composition consists of all possible interleavings of the component processes, and all of these interleavings should be described in the modal characterization.

It is also necessary to close the language on subconjunctions. Indeed, take $\hmo = \{ \diam{a}\true \wedge \diam{b} \true \}$, a language which does not meet this condition. We have $a{\bf 0} \hmeq b{\bf 0}$, but $a{\bf 0}||b{\bf 0}\models\diam{a}\true \wedge \diam{b} \true$ while $b{\bf 0}||b{\bf 0}$ does not satisfy this formula.

In case of general $\hml$ formulas, we first define a generalization of a subsequence for an arbitrary formula $\varphi \in \hml$ by specifying a set of subformulas with possible replacement from a lower level. We thus define $Sub(\varphi)$ as the smallest set of $\hml$ formulas satisfying:
\begin{itemize}
\item
$\varphi \in Sub(\varphi)$;
\item
$\varphi' \in Sub(\varphi) \implies \{ D[\psi] ~|~ \varphi' = D[C[\psi]] \} \subseteq Sub(\varphi)$.
\end{itemize}

We now define a tool to infer satisfaction of modal formulas by a parallel composition $p||q$ from the formulas satisfied by the component processes $p$ and $q$. This is accomplished by the function $Par$, which given $A\subseteq \hml$ and $B\subseteq \hml$, returns the collection of formulas that are certainly satisfied by a parallel composition of two processes satisfying $A$ and $B$ respectively. One can view $Par(A,B)$ as parallel composition operator on collections of modal formulas.

Formally, $Par: {\cal P}(\hml) \times {\cal P}(\hml) \longrightarrow {\cal P}(\hml)$ is defined with induction on the structure of formulas.
\begin{itemize}
\item $\true \in Par(A,B)$ 
\item $\diam{a} \varphi \in Par(A,B) $
$~\iffdef~$ $( \exists \diam{a} \varphi_A \in A: \varphi \in Par(\varphi_A,B))$
$\vee ( \exists \diam{a} \varphi_B \in B: \varphi \in Par(A,\varphi_B))$
\item $ \bigwedge_{i \in I} \varphi_i \in Par(A,B)$
$~\iffdef~ \forall_{i \in I}: \varphi_i \in Par(A,B)$
\item $\neg \varphi \in Par(A,B)$
$ ~\iffdef~ \forall C,D \subseteq Sub(\varphi): \varphi \in Par(C,D) ~ (\exists \psi_C \in C: \neg \psi_C \in A) \vee (\exists \psi_D \in D: \neg \psi_D \in B)$
\end{itemize}

By abuse of notation, we let $A\subseteq\hml$ also denote the formula $\bigwedge_{\varphi\in A}\varphi$.

\begin{lemma} Let $\varphi \in \hml$.
\label{lem:par}
\begin{center}
$p || q \models \varphi \iff \exists A,B \subseteq Sub(\varphi): (p \models A \wedge q \models B \wedge \varphi \in Par(A,B))$.
\end{center}
\end{lemma}

\begin{proof} 
 We use induction on the structure of formulas. The base case ($\varphi = \true$) is immediate. We proceed with the inductive step:
 
$"\implies"$: Assume that $p || q \models \varphi$. We prove that $\exists A,B \subseteq Sub(\varphi): (p \models A \wedge q \models B \wedge \varphi \in Par(A,B))$.

\begin{itemize}
\item $\varphi = \diam{a} \psi$: Without loss of generality, suppose $\lmerge{p}{q} \models \diam{a} \psi$ (the case $\lmerge{q}{p} \models \diam{a} \psi$ is symmetric), so $p \transa p' \wedge p' || q \models \psi$. From the inductive hypothesis we know that there are $A',B' \subseteq Sub(\psi)$ such that $p' \models A' \wedge q \models B' \wedge \psi \in Par(A',B')$. We take $A = \diam{a}(\bigwedge_{\varphi\in A'}\varphi)$ and $B =B'$.

\item $\varphi = \bigwedge_{i\in I} \varphi_i$: By the inductive hypothesis, for each $i\in I$ there are $A_i,B_i\subseteq Sub(\varphi_i)$ such that $p\models A_i \wedge q\models B_i \wedge \varphi_i\in Par(A_i,B_i)$. We can take $A = \bigcup_{i \in I} A_i$ and $B = \bigcup_{i \in I} B_i$.

\item $\varphi = \neg \psi$: We have:\\
$p ||q \models \neg \psi $ $\iff$ $ \neg (p ||q \models \psi)$\\
$\iff \neg (\exists C,D \subseteq Sub(\psi): (p \models C \wedge q \models D \wedge \psi \in Par(C,D)))$ (inductive hypothesis)\\
$\iff \forall C,D \subseteq Sub(\psi): \psi \in Par(C,D) 
(\exists \psi_C \in C: p \not \models \psi_C) \vee (\exists \psi_D \in D: q \not \models \psi_D)$.\\
We define:\\
$ A_p(\psi) = \bigcup_{C,D \subseteq Sub(\psi): \psi \in Par(C,D) } \{ \neg \psi_C ~|~ p \not \models \psi_C \wedge \psi_C \in C \} $\\
$ B_q(\psi) = \bigcup_{C,D \subseteq Sub(\psi): \psi \in Par(C,D) } \{ \neg \psi_D ~|~ q \not \models \psi_D \wedge \psi_D \in D \} $\\
These are the $A$ and $B$ we are looking for.
\end{itemize}

$"\Leftarrow"$: Suppose that $\exists A,B \subseteq Sub(\varphi): (p \models A, q \models B \wedge \varphi \in Par(A,B))$. We prove that $p||q \models \varphi$.
\begin{itemize}

\item $\varphi = \diam{a} \psi$: From $\diam{a} \psi \in Par(A,B)$ we have $( \exists \diam{a} \psi_A \in A: \psi \in Par(\psi_A,B))$
$\vee ( \exists \diam{a} \psi_B \in B: \psi \in Par(A,\psi_B))$. Without loss of generality suppose that $(\exists \diam{a} \psi_A \in A: \psi \in Par(\psi_A,B))$. Then $p \transa p': p' \models \psi_A$. From $\psi \in Par(\psi_A,B)$ and the inductive hypothesis we have $p' || q \models \psi$. Since $p || q \transa p' || q$, we finally obtain $p||q \models \diam{a} \psi$.

\item $\varphi = \bigwedge_{i\in I} \varphi_i$: According to the definition of $Par$ we have $\forall_{i \in I}: \varphi_i \in Par(A,B)$. The inductive hypothesis yields $\forall_{i \in I}: p||q \models \varphi_i$, and hence $p||q \models \bigwedge_{i\in I} \varphi_i$.

\item $\varphi = \neg \psi$:  We have $\forall C,D \subseteq Sub(\psi): \psi \in Par(C,D) ~ (\exists \psi_C \in C: \neg \psi_C \in A) \vee (\exists \psi_D \in D: \neg \psi_D \in B)$. Suppose, towards a contradiction, that $p||q \models \psi$. Then according to the inductive hypothesis there exist $C,D$ such that $p \models C$, $q \models D$ and $\psi \in Par(C,D)$. But from the earlier remark, we have either $\in C: \neg \psi_C \in A$ or $\psi_D \in D: \neg \psi_D \in B$. This contradicts the fact that $p \models A,C$ and $q \models B,D$.

\end{itemize}
\end{proof}

\begin{theorem} For any language $\hmo \subseteq \hml$ satisfying 
\begin{center}
(PAR)~~~~~$\varphi \in \hmo \implies Sub(\varphi) \subseteq \hmo^\equiv$,
\end{center}
 $\hmeq$ is a congruence with respect to parallel composition $||$.
\end{theorem}
\begin{proof}
Suppose $p_1 \hmeq q_1$ and $p_2 \hmeq q_2$. Suppose that $p_1 || p_2 \models \varphi \in \hmo$. According to Lem.~\ref{lem:par}, there exist $ A,B \subseteq Sub(\varphi)$ such that $p_1 \models A, p_2 \models B$ and $\varphi \in Par(A,B)$. Since $\varphi \in \hmo$, by condition PAR, $A,B\subseteq\hmo^\equiv$. Since $p_1 \hmeq q_1$ and $p_2 \hmeq q_2$, it follows that $q_1 \models A$ and $q_2 \models B$. According to Lem.~\ref{lem:par} this implies $q_1 || q_2 \models \varphi$.

\end{proof}

As an example, if we want to define a modal language that would be a congruence with respect to parallel composition, which includes behaviour described by a formula $\diam{a}(\neg \diam{b} \true \wedge \diam{c} \diam{d}\true)$, we should include the following formulas in the characterization (we omit irrelevant formulas like $\diam{a}\neg \true$): 
$\diam{a}\true$,
$\diam{a}\neg \diam{b} \true$,
$\diam{a}\diam{c} \true$,
$\diam{a}\diam{c} \diam{d} \true$,
$\diam{a}\diam{d} \true$,
$\diam{a}(\neg \diam{b} \true \wedge \diam{c}\true)$
$\diam{a}(\neg \diam{b} \true \wedge \diam{d}\true)$.

All basic equivalences except for completed trace have modal characterizations that satisfy the condition PAR. We note that parallel composition is compositional with respect to completed trace equivalence.

\section{Conclusions and future work}

We have presented, for a number of process operators from the literature, general conditions that guarantee congruence of process equivalences defined by means of a modal characterization. To the best of our knowledge it is the first such attempt.

Our conditions are sufficient, but by no means necessary. We believe that it is difficult (if not impossible) to provide a syntactic restriction on a modal language that would characterize the class of congruences for a given operator (strictly speaking, languages that induce congruences). We aimed at clear and comprehensible rather than slightly relaxed but more complicated conditions.

As the next step, we would like to investigate other process operators (e.g.\ sequential composition, renaming, merge with communication), consider the setting of weak semantics and different modal languages. In the last case, if we consider e.g.\ $\hml$ with recursion or the $\mu$-calculus, we may attempt to combine our work with existing results on characteristic formulas \cite{AcInSa09}. In that setting, instead of modal language properties, we could focus on compositionality of single formulas.

\end{document}